\documentclass[11pt]{article}
\usepackage{fullpage}

\usepackage{times}
\usepackage{comment,amsfonts,amssymb,amsmath,amsthm,graphicx,algorithmic,algorithm}
\newcommand{\commentout}[1]{}

\ifx\pdftexversion\undefined
\usepackage[colorlinks,linkcolor=black,filecolor=black,citecolor=black,urlco
lor=black,pdfstartview=FitH]{hyperref}
\else
\usepackage[colorlinks,linkcolor=blue,filecolor=blue,citecolor=blue,urlcolor
=blue,pdfstartview=FitH]{hyperref}
\fi

\newcommand{\alert}[1]{\textbf{\color{red}
[[[#1]]]}\marginpar{\textbf{\color{red}**}}\typeout{ALERT:
\the\inputlineno: #1}}

\newcommand{\R}{\mathbb{R}}

\newcommand{\eps}{\varepsilon}

\newcommand{\mommit}[1]{}
\newcommand{\namedref}[2]{\hyperref[#2]{#1~\ref*{#2}}}

\newcommand{\theoremref}[1]{\namedref{Theorem}{#1}}

\newcommand{\figureref}[1]{\namedref{Figure}{#1}}
\newcommand{\algref}[1]{\namedref{Algorithm}{#1}}

\newcommand{\lemmaref}[1]{\namedref{Lemma}{#1}}

\newcommand{\propref}[1]{\namedref{Proposition}{#1}}
\newcommand{\conjref}[1]{\namedref{Conjecture}{#1}}
\newcommand{\obref}[1]{\namedref{Observation}{#1}}

\newtheorem{theorem}{Theorem}
\newtheorem{lemma}{Lemma}

\newtheorem{claim}[lemma]{Claim}
\newtheorem{proposition}[lemma]{Proposition}

\newtheorem{conjecture}{Conjecture}
\newtheorem{definition}{Definition}
\newtheorem{observation}[lemma]{Observation}

\usepackage{pdfsync}

\begin{document}

\title{Light Spanners}
\author{
Michael Elkin\thanks{Department of Computer Science, Ben-Gurion University of the Negev,
Beer-Sheva, Israel. Email: \texttt{elkinm@cs.bgu.ac.il}}
\and
Ofer Neiman\thanks{Department of Computer Science, Ben-Gurion University of the Negev,
Beer-Sheva, Israel. Email: \texttt{neimano@cs.bgu.ac.il}. Supported in part by ISF grant No. (523/12) and by the European Union's Seventh Framework Programme (FP7/2007-2013) under grant agreement $n^\circ$303809.}
\and
Shay Solomon\thanks{Department of Computer Science and Applied Mathematics, The Weizmann Institute of Science, Rehovot 76100, Israel.
Email: \texttt{shay.solomon@weizmann.ac.il}. This work is supported by the Koshland Center for basic Research.}
}
\date{}
\maketitle
\begin{abstract}
A \emph{$t$-spanner} of a weighted undirected graph $G=(V,E)$, is a subgraph $H$ such that $d_H(u,v)\le t\cdot d_G(u,v)$ for all $u,v\in V$. The sparseness of the spanner can be measured by its size (the number of edges) and weight (the sum of all edge weights), both being important measures of the spanner's quality -- in this work we focus on the latter.

Specifically, it is shown that for any parameters $k\ge 1$ and $\eps>0$, any weighted graph $G$ on $n$ vertices admits a $(2k-1)\cdot(1+\eps)$-stretch spanner of weight at most $w(MST(G))\cdot O_\eps(kn^{1/k}/\log k)$, where $w(MST(G))$ is the weight of a minimum spanning tree of $G$.
Our result is obtained via a novel analysis of the classic greedy algorithm, and improves previous work by a factor of $O(\log k)$.
\end{abstract}

\section{Introduction}

Given a weighted connected graph $G=(V,E)$ with $n$ vertices and $m$ edges, let $d_G$ be its shortest path metric. A $t$-spanner $H=(V,E')$ is a subgraph that preserves all distances up to a multiplicative factor $t$. That is, for all $u,v\in V$, $d_H(u,v)\le t\cdot d_G(u,v)$. The parameter $t$ is called the {\em stretch}. There are several parameters that have been studied in the literature that govern the quality of $H$, two of the most notable ones are the {\em size} of the spanner (the number of edges) and its total {\em weight} (the sum of weights of its edges).

There is a basic tradeoff between the stretch and the size of a spanner. For any graph on $n$ vertices, there exists a $(2k-1)$-spanner with $O(n^{1+1/k})$ edges \cite{ADDJS93}. Furthermore, there is a simple greedy algorithm for constructing such a spanner, which we shall refer to as the {\em greedy spanner} (see \algref{alg:greedy}). The bound on the number of edges is known to be asymptotically tight for certain small values of $k$, and for all $k$ assuming Erd\H{o}s' girth conjecture.

In this paper we focus on the weight of a spanner. Light weight spanners are particularly useful for efficient broadcast protocols in the message-passing model of distributed computing \cite{ABP90,ABP91},
where efficiency is measured with respect to both the total communication cost (corresponding to the spanner's weight) and the speed of message delivery at all destinations (corresponding to the spanner's stretch).
Additional applications of light weight spanners in distributed systems include network synchronization and computing global functions \cite{ABP90,ABP91,Peleg00}.
Light weight spanners were also found useful for various data gathering and dissemination tasks in overlay networks \cite{BKRCV02,KV01},
in wireless and sensor networks \cite{SS10},  for network design \cite{MP98,SCRS01}, and routing \cite{WCT02}.

While a minimum spanning tree (MST) has the lowest weight among all possible connected spanners, its stretch can be quite large. Nevertheless, when measuring the weight of a spanner, we shall compare ourselves to the weight of an MST:
The {\em lightness} of the spanner $H$ is defined as $\frac{w(H)}{w(MST)}$ (here $w(H)$ is the total edge weight of $H$). It was shown by \cite{ADDJS93} that the lightness of the greedy spanner is at most $O(n/k)$, and their result was improved by \cite{CDNS92}, who showed that for any $\eps>0$ the greedy $(2k-1)\cdot(1+\eps)$-spanner has $O_\eps(n^{1+1/k})$ edges and lightness $O(k\cdot n^{1/k}/\eps^{1+1/k})$.
A particularly interesting special case arises when $k \approx \log n$. Specifically, in this case the result of \cite{CDNS92} provides stretch and lightness both bounded by $O(\log n)$. Another notable point on the tradeoff curve of \cite{CDNS92} (obtained by setting $\eps=\log n$ as well) is stretch $O(\log^2 n)$ and lightness $O(1)$.

These results of \cite{CDNS92} remained the state-of-the-art for more than twenty years.
In particular, prior to this work it was unknown if spanners with stretch $O(\log n)$ and lightness $o(\log n)$,
or vice versa, exist.
In this paper we answer this question in the affirmative, and show in fact something stronger -- spanners
with stretch and lightness both bounded by $o(\log n)$ exist.
We provide a novel analysis of the classic greedy algorithm, which improves the tradeoff of \cite{CDNS92} by a factor of $O(\log k)$.
Specifically, we prove the following theorem.

\begin{theorem}\label{thm:main}
For any weighted graph $G=(V,E)$ and parameters $k\ge 1$, $\eps>0$, there exists a $(2k-1)\cdot(1+\eps)$-spanner $H$ with $O(n^{1+1/k})$ edges\footnote{In fact for large $\eps$ a better bound can be obtained. Specifically, it is $O(n^{1+1/\lfloor\lceil(2k-1)\cdot(1+\eps)\rceil/2\rfloor})$.} and lightness $O(n^{1/k}\cdot(1+k/(\eps^{1+1/k}\log k)))$.
\end{theorem}

By substituting $k\approx\log n$ we obtain stretch $\log n$ and lightness $O(\log n/\log\log n)$ (for fixed small $\eps$). We also allow $\eps$ to be some large value.
In particular, setting $\eps=\log n/\log\log n$ yields stretch $\log^2n/\log\log n$ and lightness $O(1)$. Also, by substituting $k = \log n/\log\log\log n$ we can have both stretch and lightness
bounded by $O(\log n/\log\log\log n)$.

Our result shows that the potentially natural tradeoff between stretch $2k-1$ and lightness $O(k\cdot n^{1/k})$ is not the right one.
This can also be seen as an indication that the right tradeoff is stretch $(2k-1)$ and lightness $O(n^{1/k})$.
(Note that lightness $O(n^{1/k})$ is the weighted analogue of $O(n^{1+1/k})$ edges, and so
it is asymptotically tight assuming Erd\H{o}s' girth conjecture.)

\subsection{Proof Overview}

The main idea in the analysis of the greedy algorithm by \cite{CDNS92}, is to partition the edges of the greedy spanner to scales according to their weight, and bound the contribution of edges in each scale separately. For each scale they create a graph from the edges selected by the greedy algorithm to the spanner, and argue that such a graph has high girth\footnote{The girth of a graph is the minimal number of edges in a cycle.} and thus few edges. The main drawback is that when analyzing larger weight edges, this argument ignores the smaller weight edges that were already inserted into the spanner.

We show that one indeed can use information on lower weight edges when analyzing the contribution of higher scales. We create a different graph from edges added to the spanner, and argue that this graph has high girth. The new ingredient in our analysis is that we add multiple edges per spanner edge, proportionally to its weight. Specifically, these new edges form a {\em matching} between certain neighbors of the original edge's endpoints.
Intuitively, a high weight edge enforces strong restrictions on the length of cycles containing it, so it leaves a lot of "room" for low weight edges in its neighborhood. The structure of the matching enables us to exploits this room, while maintaining high girth.

Unfortunately, with our current techniques we can only use edges of weight at most $k$ times smaller than the weight of edges in the scale which is now under inspection.
Hence this gives an improvement of $O(\log k)$ to the lightness of the greedy spanner. We hope that a refinement of our method, perhaps choosing the matching more carefully, will eventually lead to an optimal lightness of $O(n^{1/k})$.

\subsection{Related Work}

A significant amount of research attention was devoted to constructing light and sparse spanners for
Euclidean and doubling metrics. A major result is that for any constant-dimensional Euclidean metric and any $\eps>0$, there exists a $(1+\eps)$-spanner with lightness $O(1)$ \cite{DHN93}. Since then there has been a flurry of work on improving the running time and other parameters. See, e.g., \cite{CDNS92,ADMSS95,DES08,ES13-focs,CLNS13}, and the references therein. An important question still left open is whether the $O(1)$ lightness
bound of \cite{DHN93} for constant-dimensional Euclidean metrics can be extended to doubling metrics. Such a light spanner has implications for the running time of a PTAS for the traveling salesperson problem (TSP). Recently, \cite{GS14} showed such a spanner exists for snowflakes\footnote{For $0\le\alpha\le 1$, an $\alpha$-snowflake of a metric is obtained by taking all distances to power $\alpha$.} of doubling metrics.

Light spanners with $(1+\eps)$ stretch have been sought for other graph families as well, with the application to TSP in mind. It has been conjectured that graphs excluding a fixed minor have such spanners. Currently, some of the known results are for planar graph \cite{ADMSS95}, bounded-genus graphs \cite{G00}, unit disk graphs \cite{KPX08}, and bounded pathwidth graphs \cite{GH12}.

A lot of research focused on constructing sparse spanners efficiently, disregarding their lightness.
Cohen \cite{Coh93} devised a randomized algorithm for constructing $((2k-1)\cdot (1+\eps))$-spanners
with $O(k \cdot n^{1+1/k} \cdot (1/\eps) \cdot \log n)$ edges. Her algorithm requires expected $O(m \cdot n^{1/k} \cdot k \cdot (1/\eps) \cdot \log n))$ time.
Baswana and Sen \cite{BS03} improved Cohen's result, and devised an algorithm that constructs $(2k-1)$-spanners with
expected $O(k \cdot n^{1+1/k})$ edges, in expected $O(k \cdot m)$ time.
Roditty et al.\ \cite{RTZ05} derandomized this algorithm, while maintaining the same parameters (including running time).
Roditty and Zwick \cite{RZ04} devised a deterministic algorithm for constructing $(2k-1)$-spanners with
$O(n^{1+1/k})$ edges in $O(k \cdot n^{2+1/k})$ time.


\section{Preliminaries}

Let $G=(V,E)$ be a graph on $n$ vertices with weights $w:E\to\R_+$, and let $d_G$ be the shortest path metric induced by $G$. For simplicity of the presentation we shall assume that the edge weights are positive integers.
(The extension of our proof to arbitrary weights is not difficult, requiring only a few minor adjustments.) For a subgraph $H=(V',E')$ define $w(H)=w(E')=\sum_{e\in E'}w(e)$. A subgraph $H=(V,E')$ is called a \emph{$t$-spanner} if for all $u,v\in V$, $d_H(u,v)\le t\cdot d_G(u,v)$. Define the {\em lightness} of $H$ as $\frac{w(H)}{w(MST(G))}$, where $MST(G)$ is a minimum spanning tree of $G$.
The girth $g$ of a graph is the minimal number of edges in a cycle of $G$. The following standard Lemma is implicit in \cite{bolo-book}.
\begin{lemma}\label{lem:girth}
Let $g>1$ be an integer. A graph on $n$ vertices and girth $g$ has at most $O\left(n^{1+\frac{1}{\lfloor (g-1)/2 \rfloor}}\right)$ edges.
\end{lemma}

\subsection{Greedy Algorithm}

The natural greedy algorithm for constructing a spanner is described in \algref{alg:greedy}.

\begin{algorithm}
\caption{$\texttt{Greedy}(G=(V,E),t)$}\label{alg:greedy}
\begin{algorithmic}[1]
\STATE $H=(V,\emptyset)$.
\FOR {each edge $\{u,v\}\in E$, in non-decreasing order of weight,}
\IF {$d_H(u,v)>t\cdot w(u,v)$}
\STATE Add the edge $\{u,v\}$ to $E(H)$.
\ENDIF
\ENDFOR
\end{algorithmic}
\end{algorithm}

Note that whenever an edge $e\in E$ is inserted into $E(H)$, it cannot close a cycle with $t+1$ or less edges, because the edges other than $e$ of such a cycle will form a path of length at most $t\cdot w(e)$ (all the existing edges are not longer than $w(e)$). This argument suggests that $H$ (viewed as an unweighted graph) has girth $t+2$ (when $t$ is an integer), and thus by \lemmaref{lem:girth}
\begin{equation}\label{eq:num}
|E(H)|\le O\left(n^{1+\frac{1}{\lfloor (t+1)/2 \rfloor}}\right)~.
\end{equation}

We observe that the greedy algorithm must select all edges of an MST (because when inspected they connect different connected components in $H$).
We will assume without loss of generality that the graph $G$ has a unique MST,
since any ties can be broken using lexicographic rules.
\begin{observation}\label{ob:mst}
If $Z$ is the MST of $G$, then $Z\subseteq H$. Furthermore, each edge in the MST does not close a cycle in $H$ when it is inspected.
\end{observation}

\section{Proof of Main Result}

Let $H$ be the greedy spanner with parameter $t=(2k-1)\cdot(1+\eps)$.
Let $Z$ be the MST of $G$, and order the vertices $v_1,v_2,\dots,v_n$ according to the order they are visited in some preorder traversal of $Z$ (with some fixed arbitrary root). Since every edge of $Z$ is visited at most twice in such a tour,
\[
L:=\sum_{i=2}^nd_Z(v_{i-1},v_i)\le 2w(Z)~.
\]

Let $I=\lceil\log_kn\rceil$. For each $i\in [I]$, define $E_i=\{e\in E(H)\setminus E(Z)\mid w(e)\in(k^{i-1},k^i]\cdot L/n\}$. We may assume the maximum weight of an edge in $H$ is bounded by $w(Z)$ (in fact $w(Z)/t$, as heavier edges surely will not be selected for the spanner), so each edge in $H \setminus Z$ of weight greater than $L/n$ is included in some $E_i$. The main technical theorem is the following.
\begin{theorem}\label{thm:tech}
For each $i\in[I]$  and any $\eps >0$,
\[
w(E_i)\le O(L\cdot (n/k^{i-1})^{1/k}/\eps^{1+1/k})~.
\]
\end{theorem}

Given this, the proof of \theoremref{thm:main} quickly follows.

\begin{proof}[Proof of \theoremref{thm:main}]
Using that the stretch of the spanner is $t\ge 2k-1$, by \eqref{eq:num} we have $|E(H)|\le O(n^{1+1/k})$.
The total weight of edges in $H$ that have weight at most $L/n$ can be bounded by $L/n\cdot |E(H)|\le L/n\cdot O(n^{1+1/k})=O(w(MST)\cdot n^{1/k})$. The contribution of the other (non-MST) edges to the weight of $H$, using \theoremref{thm:tech}, is at most
\begin{eqnarray*}
\sum_{i=1}^IO(L\cdot (n/k^{i-1})^{1/k}/\eps^{1+1/k}) &\le& O(L\cdot n^{1/k}/\eps^{1+1/k})\sum_{i=0}^\infty e^{-(i\ln k)/k}\\
&=& O(L\cdot n^{1/k}/\eps^{1+1/k})\cdot\frac{1}{1-e^{-(\ln k)/k}}\\
&=&O(w(MST))\cdot kn^{1/k}/(\eps^{1+1/k}\ln k)~.
\end{eqnarray*}
\end{proof}

\subsection{Proof of \theoremref{thm:tech}}

\paragraph{Overview:} Fix some $i\in[I]$. We shall construct a certain graph $K$ from the edges of $E_i$, and argue that this graph has high girth, and therefore few edges. The main difference from \cite{CDNS92} is that our construction combines into one scale edges whose weight may differ by a factor of $k$
(in the construction of \cite{CDNS92} all edges in a given scale are of the same weight, up to a factor of 2).
In order to compensate for heavy edges, the weight of the edge determines how many edges are added to $K$. Specifically, if the edge $\{u,v\}\in E_i$ has weight $w\cdot k^{i-1}\cdot L/n$, we shall add (at least) $\lceil w\rceil$ edges to $K$ that form a {\em matching} between vertices in some neighborhoods of $u$ and $v$. In this way the weight of $K$ dominates $w(E_i)$. To prove that $K$ has high girth, we shall map a cycle in $K$ to a closed tour in $H$ of proportional length. The argument uses the fact that the new edges are close to the original edge, and that a potential cycle in $K$ cannot exploit more than one such new edge, since these edges form a matching.

\paragraph{Construction of the Graph $K$:}
Let $P=(p_0,\dots,p_L)$ be the unweighted path on $L+1$ vertices, created from $V$ by placing $v_1,\dots,v_n$ in this order and adding Steiner vertices so that all consecutive distances are $1$, and for all $2\le j\le n$, $d_P(v_{j-1},v_j)=d_Z(v_{j-1},v_j)$. In particular, $p_0=v_1$, $p_L=v_n$, and for every $1\le j<j'\le n$,
\[
d_P(v_j,v_{j'})=\sum_{h=j+1}^{j'}d_Z(v_{h-1},v_h)~.
\]
Note that $d_P(v_j,v_{j'})\ge d_Z(v_j,v_{j'})\ge d_G(v_j,v_{j'})$, and all the inequalities may be strict. In order to be able to map edges of $K$ back to $H$, we shall also add corresponding Steiner points to the spanner $H$: For every Steiner point $p_h$ that lies on $P$ between $v_{j-1}$ and $v_j$, add a Steiner point on the path in the MST $Z$ that connects $v_{j-1}$ to $v_j$ at distance $d_P(v_{j-1},p_h)$ from $v_{j-1}$ (unless there is a point there already). By \obref{ob:mst} all MST edges are indeed in $H$, and one can simply subdivide the appropriate edge on the MST path. Note that distances in $H$ do not change, as the new Steiner points have degree $2$.
Denote by $\hat{H}$ the modified spanner $H$, i.e., $H$ with the Steiner points.

Let $a=k^{i-1}\cdot L/n$ be a lower bound on the weight of edges in $E_i$. Divide $P$ into $s=8L/(\eps a)$ intervals $I_1,\dots,I_s$, each of length $L/s= \frac{\eps}{8}a$ (by appropriate scaling, we assume all these are integers). For $j\in[s]$, the interval $I_j$ contains the points $p_{(j-1)L/s},\dots,p_{jL/s}$. In each interval $I_j$ pick an arbitrary (interior) point $r_j$ as a representative, and let $R$ be the set of representatives. For each representative $r_j$ and an integer $b\ge 0$ we define its neighborhood $N_b(j)=\{r_h ~:~ |j-h|\le b\}$
to be the set of (at most) $2b+1$ representatives that are at most $b$ intervals away from $I_j$.
(Note that the size of the neighborhood $N_b(j)$ can be smaller than $2b+1$ if $r_j$ is too close to one of the endpoints of the path $P$.)
Define an unweighted (multi) graph $K=(R,F)$ in the following manner. Let $e=\{u,v\}\in E_i$. Assume that $u\in I_h$ and $v\in I_j$ for some $h,j\in [s]$. Let $b=\lfloor w(e)/a\rfloor$, and let $M$ be an arbitrary maximal matching between $N_b(h)$ and $N_b(j)$. Add all the edges of $M$ to $F$,
see \figureref{fig:con}. For each of the edges $\{q,q'\}\in M$ added to $F$, we say that the edge $\{u,v\}$ is its {\em source} when $q\in N_b(h)$ and $q'\in N_b(j)$, and write $S(q,q')=(u,v)$.
We will soon show (in \propref{proposition:simple} below) that each edge in $K$ has a single source.

\begin{figure}[ht]
    \fbox{\includegraphics[width=460pt,clip=true,trim=0 150 30 100]{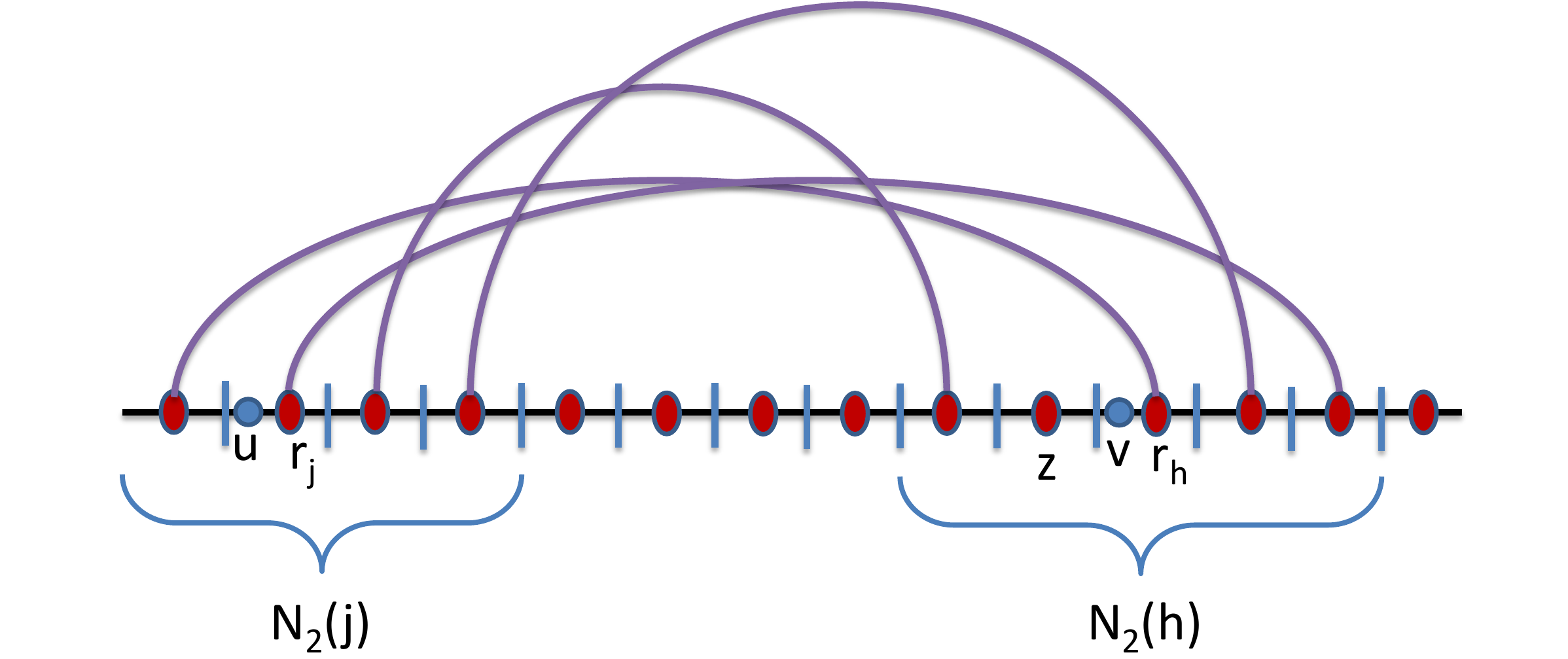}}
    \caption{{\em Construction of the graph $K$}: The oval vertices are $R$, the representatives. The edge $\{u,v\}$ is an edge of weight $2a$ selected for the spanner, and $u\in I_j$, $v\in I_h$ with representatives $r_j,r_h$. The depicted edges, that form a maximal matching between the neighborhoods of $r_j$ and $r_h$, are added to $K$. (The vertex $z$ in $N_2(h)$ does not participate in the matching, because $r_j$ is too close to the left endpoint of the path $P$.)}
\label{fig:con}
\end{figure}

The following observation suggests that if all the edges of $K$ were given weight $a$, then its total weight is greater than or equal to the weight of the edges in $E_i$.
\begin{observation}\label{ob:K-heavy}
$|F|\cdot a\ge w(E_i)$.
\end{observation}
\begin{proof}
Note that always $|N_b(j)|\ge b+1$, which means that we add at least $b+1\ge w(e)/a$ edges to $K$ for each edge $e\in E_i$. Summing over all edges concludes the proof.
\end{proof}

\noindent{\bf Mapping from $K$ to $\hat{H}$:} We shall map every edge $\{q,q'\}\in F$ to a tour $T(q,q')$ in the spanner $\hat{H}$ connecting $q$ and $q'$. If $S(q,q')=(u,v)$, then $T(q,q')$ consists of the following paths:
\begin{itemize}
\item A path in $Z$ connecting $q$ to $u$.
\item The edge $\{u,v\}$.
\item A path in $Z$ connecting $v$ to $q'$.
\end{itemize}
The following proposition asserts that the length of the tour is not longer than the weight of the source edge, up to a $1+\eps/2$ factor.
\begin{proposition}\label{proposition:map}
If an edge $\{q,q'\}\in F$ has a source $S(q,q')=(u,v)$ of weight $w$, then $T(q,q')$  is a tour in $\hat{H}$ of length at most $(1+\eps/2)w$.
\end{proposition}
\begin{proof}
First observe that the distance in $P$ between any two points in intervals $I_j$ and $I_{j+b}$ is at most $(b+1)L/s$. Since $d_P\ge d_Z$ we also have that the distance in the MST $Z$ between two such points is bounded by $(b+1)L/s$. (By definition, this holds for Steiner points as well.) Denote
 the representatives of $u,v$ as $r_j,r_h$, respectively. For $b=\lfloor w/a\rfloor$, the set $N_b(j)$ contains representatives of at most $b$ intervals away from $I_j$. As $u\in I_j$ we get that $d_Z(q,u)\le(b+1)L/s$. Similarly $d_Z(q',v)\le(b+1)L/s$, thus the total length of the tour is at most $w+2(b+1)L/s     = w + 2(\lfloor {w \over a} \rfloor + 1){\eps \over 8} a \le(1+\eps/2)w$.
\end{proof}

Our goal is to show that $K$ is a simple graph of girth at least $2k+1$. As a warmup, let us first show that $K$ does not have parallel edges.

\begin{proposition}\label{proposition:simple}
The graph $K$ does not have parallel edges.
\end{proposition}
\begin{proof}
Seeking contradiction, assume there is an edge $\{q,q'\}\in F$ with two different sources $\{u,v\},\{u',v'\}\in E_i$. Without loss of generality
assume that $\{u,v\}$ is the heavier edge of the two, with weight $w$. Then $\{q,q'\}$ is mapped to two tours in $\hat{H}$ connecting $q,q'$, whose total length, using \propref{proposition:map}, is at most $w(2+\eps)$. Consider the tour $\hat{T}=u\to q\to u'\to v'\to q'\to v$ in $\hat{H}$ which has total length at most $w(2+\eps)-w=w(1+\eps)$. Since the Steiner points have degree 2, they can be removed from $\hat{T}$ without increasing its length, and thus there is in $H$ a simple path $T$ from $u$ to $v$ of length at most $w(1+\eps)$.

We claim that $T$ must exist at the time the edge $\{u,v\}$ is inspected by the greedy algorithm. The edge $\{u',v'\}$ exists because it is lighter. The MST edges exist since by \obref{ob:mst} they must connect different components when inspected, while if some of them are inserted after $\{u,v\}$, at least one of them will close the cycle $T\cup\{u,v\}$. As $w(1+\eps)\le w\cdot (2k-1)(1+\eps)$, we conclude that the edge $\{u,v\}$ should not have been added to $H$, which is a contradiction.

\end{proof}

Showing that $K$ has large girth will follow similar lines, but is slightly more involved. The difficulty arises since we added multiple edges for each edge of $H$, thus a cycle in $K$ may be mapped to a closed tour in $H$ that uses the same edge $e\in E(H)$ more than once. In such a case, $e$ may not be a part of any simple cycle contained in the closed tour, and we will not be able to derive a contradiction from the greedy choice of $e$ to $H$. \footnote{In fact, this is the only reason our method improves the lightness by a factor of $\log k$ rather than the desired $k$.} To rule out such a possibility, we use the fact that the multiple edges whose source is $e$ form a matching, and that the weights are different by a factor of at most $k$. 

\begin{lemma}\label{lem:gir}
The graph $K=(R,F)$ has girth $2k+1$.
\end{lemma}
\begin{proof}
It will be easier to prove a stronger statement, that for any $j\in[s]$ and any $r,r'\in N_k(j)$, every path in $K$ between $r$ and $r'$ contains at least $2k+1$ edges. Once this is proven, we may use this with $r=r'$ to conclude that $\mathit{girth}(K) \ge 2k+1$.

Seeking contradiction, assume that there is a path $Q$ in $K$ from $r$ to $r'$ that contains at most $2k$ edges, and take the shortest such $Q$ (over all possible choices of $j$ and $r,r'$). Let $\{q,z\}\in F$ be the last edge added to $Q$, with source $S(q,z)=(x,y)$ (so that $\{x,y\}\in E_i$ is the heaviest among all the sources of edges in $Q$). We claim that no other edge in $Q$ has $\{x,y\}$ as a source. To see this, consider a case in which such an edge $\{q',z'\}\in F$ is also in $Q$ with $S(q',z')=(x,y)$. We may assume w.l.o.g that $q\notin\{r,r'\}$ (since the path $Q$ contains at least 2 edges), then by definition of the graph $K$, there exists some $j'\in[s]$ with $q,q'\in N_k(j')$ (recall that the neighborhood length $b$ always satisfies $b\le k$ by definition of $E_i$). But then the sub-path of $Q$ from $q$ to $q'$ is strictly shorter than $Q$, and connects two points in the same $k$-neighborhood. Since the edges in $K$ with $\{x,y\}$ as a source form a matching, we get that $q\neq q'$, and thus this path is not of length 0. This contradicts the minimality of $Q$. Next, we will show that $\{x,y\}$ should not have been chosen for $H$, because there is a short path connecting $x$ to $y$.

By \propref{proposition:map} every edge $e\in Q$ whose source $S(e)=e'$ has weight $w(e')$, is mapped to a tour $T(e)$ of length at most $(1+\eps/2)w(e')$ in $\hat{H}$. Since $w(x,y)$ is the maximum weight source of all edges in $Q$, we conclude that the total length of tours connecting $x$ to $r$ and $r'$ to $y$ is at most $(2k-1)\cdot (1+\eps/2)w(x,y)$. Note that $r,r'$ are representatives in $N_k(j)$, which are at most $2k$ intervals apart. So their distance in the MST $Z$ is at most $2k\cdot\eps a/8\le k\cdot\eps w(x,y)/4$. The total length of the  tour  $x\to r\to r'\to y$ in $\hat{H}$ is at most
\[
(2k-1)\cdot (1+\eps/2)w(x,y)+k\cdot\eps w(x,y)/4 \le (2k-1)(1+\eps)\cdot w(x,y)~.
\]
When the algorithm considers the edge $\{x,y\}$, all the edges of the above tour exist in $\hat{H}$. (This follows since they are all MST edges or lighter than $w(x,y)$, similarly to the argument used in \propref{proposition:simple}.)  We conclude that there is a path between $x$ and $y$ in $H$ of length at most $(2k-1)\cdot(1+\eps)\cdot w(x,y)$, hence $\{x,y\}$ should not have been added to $E(H)$, which yields a contradiction.

\end{proof}

\begin{proof}[Proof of \theoremref{thm:tech}]
Recall that the graph $K$ has $s$ vertices. By \propref{proposition:simple} it is a simple graph, and \lemmaref{lem:gir} suggests it has girth at least $2k+1$, thus by using \lemmaref{lem:girth} it has at most $O(s^{1+1/k})$ edges. Using \obref{ob:K-heavy},
\begin{eqnarray*}
w(E_i)&\le&|F|\cdot a\\
&\le& O(s^{1+1/k})\cdot(L/n\cdot k^{i-1})\\
&=&\left(\frac{8 \cdot n}{\eps k^{i-1}}\right)^{1+1/k}\cdot(O(L)/n\cdot k^{i-1})\\
&\le&O(L\cdot (n/k^{i-1})^{1/k}/\eps^{1+1/k})~.
\end{eqnarray*}
\end{proof}

\section{Weighted Girth Conjecture}

The girth of a graph is defined on unweighted graphs. Here we give an extension of the definition that generalizes to weighted graphs as well, and propose a conjecture on the extremal graph attaining a weighted girth.

\begin{definition}
Let $G=(V,E)$ be a weighted graph with weights $w:E\to\R_+$, the {\em weighted girth} of $G$ is the minimum over all cycles $C$ of the weight of $C$ divided by its heaviest edge, that is
\[
\min_{C\text{ cycle in G}}~\left\{\frac{w(C)}{\max_{e\in C}w(e)}\right\}~.
\]
\end{definition}
Note that this matches the standard definition of girth for unweighted graphs.
Recall that the lightness of $G$ is $\frac{w(G)}{w(MST)}$. For a given weighted girth value $g$ and cardinality $n$, we ask what is the graph on $n$ vertices with weighted girth $g$ that maximizes the lightness?
\begin{conjecture}\label{con:wg}
For any integer $g\ge 3$, among all graphs with $n$ vertices and weighted girth $g$, the maximal lightness is attained for an unweighted graph.
\end{conjecture}
Recall that Erd\H{o}s' girth conjecture asserts that there exists an (unweighted) graph with girth $g>2k$ and $\Omega(n^{1+1/k})$ edges, that is, its lightness is $\Omega(n^{1/k})$. Observe that any graph of weighted girth larger than $2k+\eps(2k-1)$ can be thought of as the output of \algref{alg:greedy} with parameter $t=(2k-1)\cdot(1+\eps)$. In particular, \theoremref{thm:main} implies that its lightness is at most $O_\eps(kn^{1/k}/\log k)$. Thus (up to the term of $\eps(2k-1)$ in the girth), there exists an unweighted graph which is at most $O(k/\log k) = O(g/\log g)$ lighter than the heaviest weighted graph.

The intuition behind this conjecture follows from our method of replacing high weight edges by many low weight edges. We believe that such replacement should hold when performed on all possible scales simultaneously.
An immediate corollary of \conjref{con:wg}, is that the lightness of a greedy $(2k-1)$-spanner of a weighted graph on $n$ vertices is bounded by $O(n^{1/k})$. To see why this is true, note that the spanner's weighted girth must be strictly larger than $2k$, and $O(n^{1/k})$ is a bound on the lightness of an {\em unweighted} graph on $n$ vertices with girth $2k+1$.

{\small \bibliographystyle{alpha}
\bibliography{latex8}}
\end{document}